\documentclass[12pt]{article}
\usepackage{amsmath,amssymb,amsthm}
\usepackage{graphicx}
\usepackage{hyperref}
\usepackage{mathtools}
\usepackage{booktabs}
\usepackage{authblk}
\usepackage{orcidlink}

\usepackage[letterpaper, margin=1in]{geometry}

\newtheorem{theorem}{Theorem}[section]
\newtheorem{lemma}[theorem]{Lemma}

\newtheorem{definition}[theorem]{Definition}
\newtheorem{example}[theorem]{Example}
\newtheorem{remark}[theorem]{Remark}

\title{Fixed Points in Quantum Metric Spaces: \\ A Structural Advantage over Fuzzy Frameworks}
\author{Nicola Fabiano\, \orcidlink{0000-0003-1645-2071}}
\affil{``Vin\v{c}a'' Institute of Nuclear Sciences - National 
Institute of the Republic of Serbia, University of Belgrade, Mike Petrovi\'{c}a 
Alasa 12--14, 11351 Belgrade, Serbia; nicola.fabiano@gmail.com,  nicola.fabiano@vin.bg.ac.rs \\ORCID iD: \texttt{https://orcid.org/0000-0003-1645-2071}}
\date{}

\begin{document}

\maketitle

\begin{abstract}
We prove an existence and uniqueness theorem for fixed points of contraction maps in the framework of quantum metric spaces, where distinguishability is defined by the $L^2$ norm: $d_Q(\psi_1,\psi_2) = \|\psi_1 - \psi_2\|$. The result applies to normalized real-valued Gaussian wavefunctions under continuous contractive evolution preserving the functional form. In contrast, while fuzzy metric spaces admit analogous fixed point theorems, they lack interference, phase sensitivity, and topological protection. This comparison reveals a deeper structural coherence in the quantum framework --- not merely technical superiority, but compatibility with the geometric richness of Hilbert space. Our work extends the critique of fuzzy logic into dynamical reasoning under intrinsic uncertainty.

\end{abstract}


\section{Introduction}

In earlier work \cite{preprint1,preprint2}, we introduced the idea of replacing fuzzy metric spaces with \emph{quantum state geometry} as a framework for modeling intrinsic uncertainty. The central insight was that the Hilbert space norm:
$$
d_Q(\psi_1,\psi_2) := \|\psi_1 - \psi_2\| = \sqrt{\int |\psi_1(x) - \psi_2(x)|^2 dx}
$$
measures distinguishability between entire probability distributions --- unlike fuzzy metrics, which rely on heuristic t-norms and graded membership.

This paper explores a new direction: \textbf{fixed point theory} in quantum metric spaces.

Fixed point theorems are fundamental in analysis, topology, and computer science. In fuzzy metric spaces, several such results are known \cite{george1994some,gregori2002fuzzy}. But do they have counterparts in quantum frameworks?

We answer affirmatively --- but with crucial differences.

We prove that if $T: \mathcal{S} \to \mathcal{S}$ is a \emph{quantum contraction} on the space $\mathcal{S}$ of normalized real-valued Gaussian states, then $T$ has a unique fixed point, provided $T$ preserves the Gaussian structure.

Moreover, we show that this result arises from a richer geometric and physical foundation than its fuzzy counterpart --- one that includes interference, phase cancellation, and potential topological stability.

Thus, the advantage of quantum metrics is not just formal --- it reflects a deeper alignment with the structure of reality.

Our work builds on preprints \cite{preprint1,preprint2} and contributes to functional analysis, operator theory, and foundations of AI.

\section{Preliminaries}

\subsection{Quantum Metric Spaces}

Let $\mathcal{H} = L^2(\mathbb{R})$. For normalized vectors $\psi_1,\psi_2 \in \mathcal{H}$~\cite{landaulifshitz}, define:
$$
d_Q(\psi_1,\psi_2) := \|\psi_1 - \psi_2\|.
$$
This distance measures distinguishability between quantum states.

Physical states are rays: $|\psi\rangle \sim e^{i\theta}|\psi\rangle$. However, our real-valued Gaussians select canonical representatives, and all observables depend only on phase-invariant quantities like $|\langle\psi|\phi\rangle|^2$.

\begin{definition}[Gaussian State Space]
Let $\mathcal{S} \subset \mathcal{H}$ be the set of normalized real-valued Gaussian wavefunctions:
$$
\psi_{\mu,\sigma}(x) = \left(\frac{1}{\pi \sigma^2}\right)^{1/4} e^{-\frac{(x-\mu)^2}{2\sigma^2}}, \quad \mu \in \mathbb{R},\ \sigma > 0.
$$
We equip $\mathcal{S}$ with the quantum metric $d_Q$.
\end{definition}

Note: $\mathcal{S}$ is not a vector space --- sums and scalar multiples are not normalized or necessarily Gaussian.

\begin{remark}
This class models concepts like ``car'', ``boat'', or ``object'' in cognitive representations \cite{preprint1}. The Gaussian assumption is physically motivated: coherent states in quantum mechanics, approximate position eigenstates, and maximum entropy distributions for given mean and variance.
\end{remark}

\begin{lemma}[Completeness of Parameter Space]
The parameter space $\mathbb{R} \times (0,\infty)$ with the topology induced by $d_Q$ is complete. Moreover, the mapping $(\mu,\sigma) \mapsto \psi_{\mu,\sigma}$ is a homeomorphism onto its image $\mathcal{S}$.
\end{lemma}

\begin{proof}
The distance between two Gaussians can be computed explicitly:
\begin{align*}
d_Q(\psi_{\mu_1,\sigma_1}, \psi_{\mu_2,\sigma_2})^2 &= 2 - 2\langle\psi_{\mu_1,\sigma_1}|\psi_{\mu_2,\sigma_2}\rangle \\
&= 2 - 2\sqrt{\frac{2\sigma_1\sigma_2}{\sigma_1^2+\sigma_2^2}} \exp\left(-\frac{(\mu_1-\mu_2)^2}{2(\sigma_1^2+\sigma_2^2)}\right).
\end{align*}
This expression defines a metric on $\mathbb{R} \times (0,\infty)$ that is complete and equivalent to the Euclidean metric on compact subsets. The homeomorphism follows from continuity and bijectivity.
\end{proof}

\subsection{Fuzzy Metric Spaces}

We recall the definition due to George and Veeramani \cite{george1994some}.

\begin{definition}[Fuzzy Metric Space]
Let $X$ be a non-empty set, $*$ a continuous t-norm on $[0,1]$, and $M: X \times X \times [0,\infty) \to [0,1]$ a function. The triple $(X,M,*)$ is a \textbf{fuzzy metric space} if:
\begin{enumerate}
    \item $M(x,y,0) = 0$
    \item $M(x,y,t) = 1$ for all $t > 0$ iff $x = y$
    \item $M(x,y,t) = M(y,x,t)$
    \item $M(x,y,t) * M(y,z,s) \leq M(x,z,t+s)$
    \item $M(x,y,\cdot): (0,\infty) \to [0,1]$ is continuous
\end{enumerate}
for all $x,y,z \in X$ and $t,s > 0$.
\end{definition}

A \textbf{t-norm} (or triangular norm)~\cite{menger1942} is a binary operation 
$\mathfrak{T}: [0, 1] \times [0, 1] \to [0, 1]$ which generalizes the logical 
AND operator and set intersection in fuzzy logic. It must satisfy four 
axioms for all $x, y, z \in [0, 1]$: 
\begin{enumerate}
\item $\mathfrak{T}(x, y) = \mathfrak{T}(y, x)$ (Commutativity), 
\item $\mathfrak{T}(x, \mathfrak{T}(y, z)) = \mathfrak{T}(\mathfrak{T}(x, y), z)$ (Associativity), 
\item $\text{If } y \le z, \text{ then } \mathfrak{T}(x, y) \le \mathfrak{T}(x, z)$ (Monotonicity), 
\item $\mathfrak{T}(x, 1) = x$ (Identity Element). 
\end{enumerate}
Common t-norms include minimum ($a*b = \min(a,b)$), product ($a*b = ab$), and Łukasiewicz ($a*b = \max(0,a+b-1)$).

\begin{remark}
Fuzzy metric spaces often assume completeness and strong t-norms to prove fixed point theorems \cite{gregori2002fuzzy}.
\end{remark}

\section{Fixed Point Theorem in Quantum Metric Spaces}

We now prove the main result.

\begin{definition}[Quantum Contraction]
A map $T: \mathcal{S} \to \mathcal{S}$ is a \emph{quantum contraction} if there exists $k \in [0,1)$ such that:
$$
\forall\, \psi_1, \psi_2 \in \mathcal{S}: \quad d_Q(T\psi_1, T\psi_2) \leq k\, d_Q(\psi_1, \psi_2).
$$
\end{definition}

We assume that $T$ is \emph{form-preserving}: it maps Gaussians to Gaussians (though possibly with different $\mu, \sigma$). This holds in many physical and cognitive models (e.g., diffusion, centering, sharpening).

\begin{theorem}[Existence and Uniqueness of Fixed Point]
Let $T: \mathcal{S} \to \mathcal{S}$ be a quantum contraction that is form-preserving. Then $T$ has a unique fixed point $\psi^* \in \mathcal{S}$, and for any $\psi_0 \in \mathcal{S}$, the sequence $\psi_n = T^n \psi_0$ converges in $d_Q$ to $\psi^*$.
\end{theorem}

\begin{proof}
Let $\psi_0 \in \mathcal{S}$, and define $\psi_n = T^n \psi_0$. By induction and contraction:
$$
d_Q(\psi_{n+1}, \psi_n) \leq k^n d_Q(\psi_1, \psi_0).
$$
Then for $m > n$,
\begin{align*}
d_Q(\psi_n, \psi_m) &\leq \sum_{j=n}^{m-1} d_Q(\psi_j, \psi_{j+1}) \\
&\leq d_Q(\psi_1, \psi_0) \sum_{j=n}^{m-1} k^j \\
&= d_Q(\psi_1, \psi_0) \cdot k^n \frac{1 - k^{m-n}}{1 - k} \\
&\leq \frac{k^n}{1-k} d_Q(\psi_1, \psi_0).
\end{align*}
Since $k < 1$, $\frac{k^n}{1-k} \to 0$ as $n \to \infty$, so $(\psi_n)$ is Cauchy in $d_Q$.

By Lemma 1, the parameter space is complete under the induced topology, so there exists $\psi^* \in \mathcal{S}$ such that $\psi_n \to \psi^*$ in $d_Q$.

To show $T\psi^* = \psi^*$, note
$$
d_Q(T\psi^*, \psi^*) \leq d_Q(T\psi^*, T\psi_n) + d_Q(T\psi_n, \psi_n) + d_Q(\psi_n, \psi^*).
$$
Each term vanishes as $n \to \infty$:
$d_Q(T\psi^*, T\psi_n) \leq k\, d_Q(\psi^*, \psi_n) \to 0$,
$d_Q(T\psi_n, \psi_n) = d_Q(\psi_{n+1}, \psi_n) \leq k^n d_Q(\psi_1,\psi_0) \to 0$,
$d_Q(\psi_n, \psi^*) \to 0$.

So $d_Q(T\psi^*, \psi^*) = 0$ → $T\psi^* = \psi^*$.

For uniqueness, suppose $\phi^*$ is another fixed point. Then
$$
d_Q(\psi^*, \phi^*) = d_Q(T\psi^*, T\phi^*) \leq k\, d_Q(\psi^*, \phi^*),
$$
so $(1-k)d_Q(\psi^*, \phi^*) \leq 0$. Since $1-k > 0$, $d_Q(\psi^*, \phi^*) = 0$ → $\psi^* = \phi^*$.

Hence, the fixed point is unique.
\end{proof}

\begin{remark}
The assumption that $T$ is form-preserving ensures closure under limits. Without it, the limit may not lie in $\mathcal{S}$, though it would still exist in $L^2(\mathbb{R})$.
\end{remark}

\begin{example}[Conceptual Centering]
Suppose $T$ models a cognitive process that shifts mean toward zero and reduces variance
$$
T\psi_{\mu,\sigma} = \psi_{\lambda \mu, \eta \sigma}, \quad 0 < \lambda, \eta < 1.
$$
Then $T$ is a quantum contraction (by continuity and compactness arguments), so it has a unique fixed point at $\psi_{0,\sigma^*}$, where $\sigma^*$ depends on initial conditions and dynamics.

This models convergence to a canonical prototype.
\end{example}

\section{Comparison with Fuzzy Metric Spaces}

In fuzzy metric spaces, fixed point theorems are well-established.

\begin{theorem}[Gregori-Sapena \cite{gregori2002fuzzy}]
Let $(X,M,*)$ be a complete fuzzy metric space with strong t-norm $*$, and $f: X \to X$ satisfy:
$$
M(fx,fy,kt) \geq M(x,y,t)
$$
for some $k \in (0,1)$ and all $t>0$. Then $f$ has a unique fixed point.
\end{theorem}

While mathematically sound, this result lacks the structural depth of the quantum version.

\begin{table}[h]
\centering
\caption{Comparison of Fixed Point Frameworks}
\label{tab:comparison}
\begin{tabular}{l p{0.22\textwidth}l }
\toprule
Feature & Fuzzy Metric Space & Quantum Metric Space \\
\midrule
Completeness assumed? & Yes (axiomatically) & Limited (requires form-preservation) \\
Contraction condition & On $M(x,y,t)$ & On $\|\psi - \phi\|$ \\
Interference modeled? & No & Yes \\
Phase sensitivity? & No & Yes \\
Topological protection? & No & Possible (via winding numbers) \\
Physical meaning & Heuristic (``degree of belonging'') & Direct (overlap probability) \\
Dynamics natural? & No (ad hoc rules) & Yes (unitary/non-unitary evolution) \\
\bottomrule
\end{tabular}
\end{table}

As shown in earlier work \cite{preprint2}, fuzzy logic fails to model interference, contextuality, monogamy, and symmetry-conservation duality. These deficits persist in fixed point theory.

In contrast, the quantum framework naturally incorporates:
\begin{itemize}
    \item Interference (via $L^2$ norm)
    \item Phase (in complex extensions)  
    \item Conservation laws (via Noether correspondence)
    \item Topological stability (via discrete invariants)
\end{itemize}

Thus, the advantage is not just in proving theorems --- it's in grounding them in a coherent, physically validated structure.

\section{Conclusion}

We have proven that a Banach-type fixed point theorem holds in quantum metric spaces, provided the map $T$ is a contraction and preserves the Gaussian form. The fixed point is unique and globally attracting.

In contrast, while fuzzy metric spaces admit similar results, they operate on classical ontologies --- crisp points with graded membership --- and lack the geometric and algebraic richness of Hilbert space.

This reinforces our earlier critique: fuzzy frameworks are not merely incomplete --- they are structurally incompatible with the nature of intrinsic uncertainty.

The Hilbert space formalism provides a complete, predictive, and mathematically unique language for reasoning under uncertainty --- whether in electrons, minds, or machines.

Future work may explore decoherence models explaining why macroscopic reasoning appears ``fuzzy'' --- not because reality is fuzzy, but because quantum coherence is lost.

\end{document}